\documentclass[openacc]{rsproca_new}

\usepackage{graphicx}
\usepackage{hyperref}
\usepackage{orcidlink}
\usepackage{tikz}
\usepackage{braket}

\newtheorem{theorem}{Theorem}[section]
\newtheorem{lemma}[theorem]{Lemma}
\newtheorem{proposition}[theorem]{Proposition}

\titlehead{Research}

\begin{document}

\title{Exact Quantum Maxima of the $n$-Cycle Overlap Inequalities}

\author{Mohd Asad Siddiqui\orcidlink{0000-0001-5003-7571}}

\address{Jawaharlal Nehru Rajkeeya Mahavidyalaya,
Sri Vijaya Puram 744104, Andaman and Nicobar Islands, India}

\subject{Quantum coherence, quantum contextuality, multipath interferometry}


\corres{Mohd Asad Siddiqui\\
\email{asad@ctp-jamia.res.in}}

\others{Electronic supplementary material is available online at
\url{https://doi.org/10.6084/m9.figshare.c.8739018}.}

\begin{abstract}
We extend the three-state overlap analysis to determine the exact quantum maximum over finite-dimensional pure-state realizations of the $n$-cycle overlap inequalities, $S_n^{\max}=n\cos^2(\pi/(2n))-1$, for arbitrary cycle length $n\ge3$. The bound is saturated by an explicit family of coplanar qubit states equally spaced along a Fubini--Study geodesic, establishing dimensional saturation of the overlap-cycle hierarchy. Thus, the global optimum over all finite-dimensional pure-state realizations is already achieved in dimension two. We further show how, under ideal symmetric interferometric conditions, the overlap quantities can be inferred from pairwise fringe visibilities. The three-state case recovers the known maximum $5/4$, while the exact $n$-cycle result shows that the quantum--classical gap approaches one as $n\to\infty$, with corresponding visibility thresholds. Within generalized noncontextuality frameworks, and subject to the required operational equivalences, such violations can witness preparation contextuality.
\end{abstract}

\maketitle

\section{Introduction}

Quantum coherence is a defining feature of quantum theory and a primary resource for quantum technologies. In the standard resource theory~\cite{Baumgratz2014}, coherence is defined relative to a fixed reference basis. A complementary, basis-independent notion was introduced by Designolle \emph{et al.}~\cite{Designolle2021}, which characterizes coherence as the impossibility of simultaneously diagonalizing a set of quantum states in a common basis. For triples of states with pairwise overlaps $r_{ij}=|\langle d_i|d_j\rangle|^2$, coherence-free models satisfy nontrivial overlap inequalities~\cite{Galvao2020,Wagner2024}. Violations of these inequalities witness basis-independent coherence
~\cite{Galvao2020,Wagner2024} and, when the requisite operational
equivalences are independently established, can also witness preparation
contextuality within Spekkens' generalized noncontextuality framework
~\cite{Spekkens2005,Wagner2024}. Recent work has also explored the interplay between coherence and contextuality in interferometric settings~\cite{Wagner2024MZI}.

State overlaps are often accessed through tomography or specialized
overlap-estimation protocols such as SWAP tests~\cite{Ekert2002}.
Multi-path interferometry provides a natural operational alternative:
coherence manifests through pairwise interference visibilities
$\mathcal{V}_{ij}$~\cite{Qureshi2019}, which, under ideal symmetric
conditions, directly encode the detector-state overlaps. This connection
is closely related to coherence-based formulations of multipath
wave-particle duality~\cite{Bera2015,Qureshi2017}. Related studies of
three-path interference~\cite{Siddiqui2015} and multipath wave-particle
duality with a path detector in a quantum superposition~\cite{Siddiqui2021}
have also been reported.

Under the pure-state, equal-amplitude, and ideal-interferometry conditions,
pairwise visibility measurements directly determine the corresponding quantum
overlaps of the detector states. Any coherence-free
description of a three-path interferometer must therefore satisfy
\begin{equation}
\mathcal{V}_{12}^2 + \mathcal{V}_{23}^2 - \mathcal{V}_{13}^2 \le 1,
\end{equation}
whereas quantum theory allows the larger value $5/4$, attained by pure qubit
states. We then extend the analysis to arbitrary $n$-cycle inequalities and
determine the exact quantum maximum
\begin{equation}
S_n^{\max}
=
n\cos^2\!\left(\frac{\pi}{2n}\right)-1,
\end{equation}
where $S_n=\sum_{i=1}^{n-1}r_{i,i+1}-r_{1n}$ is the $n$-cycle overlap
expression introduced in Sec.~\ref{sec:tight-n}. This value is attained by an
explicit family of pure qubit states equally spaced along a Fubini--Study
geodesic of total length $(n-1)\pi/(2n)$, equivalently a Bloch-sphere arc of
angular length $(n-1)\pi/n$. This establishes dimensional saturation for the
overlap-cycle hierarchy: the exact quantum maximum over all finite-dimensional
pure-state realizations is already achieved in dimension two, and
higher-dimensional systems offer no advantage.

Previous work established overlap-cycle inequalities and identified the
maximal quantum violation for the three-state case~\cite{Galvao2020}, while
later work developed more general graph-based classicality inequalities
~\cite{Wagner2024}. While the three-state case is exactly solvable, to our knowledge, the exact quantum maximum for arbitrary $n$-cycle overlap inequalities has not previously been reported.

Exact quantum bounds are known for several cycle contextuality scenarios,
beginning with the Klyachko--Can--Binicioğlu--Shumovsky (KCBS) inequality and its $n$-cycle generalizations
~\cite{Klyachko2008,Araujo2013}. The overlap-cycle problem considered here has
a different objective, involving pairwise state overlaps rather than
contextual correlations.

The optimization can also be viewed in the Gram-matrix and graph-theoretic
framework. For pure states, the matrix with entries
$r_{ij}=\operatorname{Tr}(P_iP_j)$ is real positive semidefinite with unit
diagonal and therefore belongs to the elliptope $\mathcal E_n$~\cite{DezaLaurent1997}.
In the present setting, the admissible matrices form a restricted subset of
$\mathcal E_n$, with $r_{ij}\ge0$ and with all entries arising from a common
set of pure states. The conventional $n$-cycle
correlation inequalities have well-known parity-dependent quantum bounds
~\cite{Araujo2013}, while the associated exclusivity-event formulation
admits a graph-theoretic description in terms of Lov\'asz-$\vartheta$-type
quantities~\cite{CabelloSeveriniWinter2014}; for example, for odd $n$, the KCBS
probability sum has maximum
$\vartheta(C_n)=n\cos(\pi/n)/[1+\cos(\pi/n)]$
~\cite{Araujo2013,CabelloSeveriniWinter2014}. More broadly, graph-theoretic
and combinatorial approaches to contextuality provide related descriptions
of quantum and noncontextual correlation sets
~\cite{CabelloSeveriniWinter2010,AmaralTerraCunha2024,Acin2015}.

The present $S_n$, however, is neither the standard $n$-cycle correlation
functional nor an identified Lov\'asz-$\vartheta$ invariant. Its $n-1$
positively weighted adjacent overlaps and single negatively weighted closing
overlap define an asymmetric extremal problem, distinct from these standard
cycle and exclusivity-event functionals.

Bargmann invariants provide a related but distinct description of quantum-state geometry through cyclic products of complex overlaps, with associated gauge-invariant phase information~\cite{MukundaBargmann2003}. The present
functional instead depends only on pairwise squared moduli
$|\langle d_i|d_j\rangle|^2$ and does not explicitly involve Bargmann phases.
Recent Bargmann-scenario and Bargmann-polytope frameworks further develop this geometric characterization of coherence in sets of states~\cite{WagnerBargmann2026}.
Experimentally, generalized contextuality tests that account for the absence of exact operational idealizations have already been demonstrated~\cite{Mazurek2016}.

This paper is organized as follows. Section~\ref{sec:interferometer} derives the visibility formulation for multipath interferometers. Section~\ref{sec:three_path} presents the three-path visibility inequality and its maximal quantum violation. Section~\ref{sec:lsss} connects the visibility inequalities to preparation noncontextuality. Section~\ref{sec:tight-n} generalizes the results to arbitrary cycle lengths and determines the exact quantum maximum. Section~\ref{sec:discussion} discusses experimental feasibility. Finally, Section~\ref{sec:conclusion} summarizes the main conclusions.

\section{Interferometer and Visibilities}
\label{sec:interferometer}

We assume standard quantum mechanics, in which interference is fully
described by pairwise terms and higher-order interference is absent.
Consider an $n$-path interferometer in which a quantum particle is
prepared in the pure state
\begin{equation}
\ket{\Psi}
=
\sum_{i=1}^{n} c_i \ket{\psi_i},
\qquad
\sum_i |c_i|^2 = 1,
\end{equation}
where $c_i$ are path amplitudes and $\{\ket{\psi_i}\}$ is an orthonormal
path basis. We restrict attention to path pairs with nonzero amplitudes, $c_i,c_j\neq0$, since a zero-amplitude path does not contribute to the corresponding two-path interference measurement. Detector states associated with paths of zero amplitude are therefore not operationally probed in that configuration.

Which-path information is extracted by coupling each path to a detector
initially prepared in the reference state $\ket{d_0}$. The interaction
correlates each path with a distinct detector state,
\begin{equation}
\ket{\psi_i}\ket{d_0}
\longrightarrow
\ket{\psi_i}\ket{d_i}.
\end{equation}
The resulting particle--detector state is
\begin{equation}
\ket{\Psi}_{QD}
=
\sum_i c_i \ket{\psi_i}\ket{d_i}.
\end{equation}

Tracing over the path degree of freedom yields the reduced detector state
\begin{equation}
\rho_D
=
\sum_i |c_i|^2 \ket{d_i}\bra{d_i}.
\end{equation}

We assume that the detector states $\ket{d_i}$ are independent of
which subset of paths is opened, so that all pairwise visibilities are
associated with a single underlying family of detector states. In an
experimental implementation, this assumption should be checked by keeping
the state-preparation optics fixed across the pair-selection configurations
and by verifying the prepared detector states (or equivalent calibration
observables) to within the required experimental uncertainty. We further
assume stable path amplitudes and ideal interferometric conditions in which
visibility loss arises solely from detector-state distinguishability.

For any pair of paths $(i,j)$, let $\mathcal{V}_{ij}$ denote the corresponding
two-path fringe visibility and define the squared detector-state overlap
$r_{ij}=|\langle d_i|d_j\rangle|^2$. Under the ideal pure-state model,~\cite{Qureshi2019}
\begin{equation}
\mathcal{V}_{ij}
=
\frac{2|c_i c_j|}
{|c_i|^2+|c_j|^2}
\,|\langle d_i|d_j\rangle|.
\label{eq:vis-general}
\end{equation}
Operationally, $r_{ij}$ is the probability of obtaining the outcome
associated with the projector $|d_i\rangle\langle d_i|$ when measuring
the state $\ket{d_j}$. While the visibility is linear in
$|\langle d_i|d_j\rangle|$, its square is proportional to $r_{ij}$ and
coincides with it in the symmetric case discussed below.

For equal path amplitudes,
\[
|c_i|=\frac1{\sqrt n}
\qquad (i=1,\ldots,n),
\]
Eq.~\eqref{eq:vis-general} simplifies to
\begin{equation}
\mathcal{V}_{ij}
=
|\langle d_i|d_j\rangle|,
\qquad
\mathcal{V}_{ij}^{\,2}
=
r_{ij}.
\label{eq:vis-sym}
\end{equation}
Thus pairwise visibility measurements directly access the overlaps
appearing in the overlap inequalities, without requiring full state
tomography or multi-copy overlap measurements such as SWAP tests. Distinct
path-pair selection configurations are nevertheless required.

The identification $\mathcal{V}_{ij}^{\,2}=r_{ij}$ assumes pure detector
states, equal path amplitudes, and ideal interferometric conditions. For mixed detector states, fringe visibility does not in general coincide
with the Hilbert--Schmidt overlap $\operatorname{Tr}(\rho_i\rho_j)$.
It should also not be identified, without specifying a physical interference
model, with the Uhlmann fidelity
$F(\rho_i,\rho_j)=\bigl[\operatorname{Tr}\sqrt{\sqrt{\rho_i}\rho_j\sqrt{\rho_i}}\bigr]^2$.
Accordingly, the pure-state optimization below should not be interpreted as
a statement about an arbitrary mixed-state visibility model.

In realistic experiments, decoherence, mode mismatch, phase drift, and other
imperfections can reduce the observed visibilities~\cite{Sharma2020}. Because the cycle
expression contains one negatively weighted term, arbitrary visibility loss
is not automatically conservative. We therefore use the uniform reduction
model only for the simple threshold estimates below and discuss calibrated
pair-dependent imperfections separately.

\section{Classical polytope and visibility inequality (three paths)}
\label{sec:three_path}

For three detector states
$\{|d_1\rangle,|d_2\rangle,|d_3\rangle\}$,
the overlap triples compatible with the coherence-free (jointly diagonalizable)
model obey the facet inequalities
\begin{equation}
r_{12}+r_{23}-r_{13}\le1,
\label{classical}
\end{equation}
together with its permutations. These inequalities characterize the
classical overlap region considered in Refs.~\cite{Galvao2020,Wagner2024}.

Equivalently, Eq.~\eqref{classical} can be rewritten as
\begin{equation}
(1-r_{13})\le(1-r_{12})+(1-r_{23}),
\label{triangle-disagreement}
\end{equation}
which takes the form of a triangle inequality for the disagreement
probabilities
\[
P_{ij}(\mathrm{disagree})=1-r_{ij}.
\]
In any jointly diagonalizable description, the disagreement probability
between states $1$ and $3$ cannot exceed the sum of the disagreement
probabilities involving the adjacent pairs $(1,2)$ and $(2,3)$.

Using Eq.~\eqref{eq:vis-general}, this constraint can be expressed
entirely in terms of measurable visibilities, yielding the
\emph{visibility inequality}
\begin{align}
&\frac{(|c_1|^2+|c_2|^2)^2}{4|c_1c_2|^2}\,\mathcal{V}_{12}^2
+\frac{(|c_2|^2+|c_3|^2)^2}{4|c_2c_3|^2}\,\mathcal{V}_{23}^2
-\frac{(|c_1|^2+|c_3|^2)^2}{4|c_1c_3|^2}\,\mathcal{V}_{13}^2
\le1,
\label{vis-ineq-asym}
\end{align}
valid for arbitrary nonzero path amplitudes under ideal interferometric
conditions. (If a path amplitude vanishes, the corresponding path is
absent from the effective cycle and the inequality is interpreted on the
reduced set of paths with nonzero amplitudes.)

For the symmetric interferometer ($|c_i|=1/\sqrt3$), this reduces to
\begin{equation}
\mathcal{V}_{12}^2+\mathcal{V}_{23}^2-\mathcal{V}_{13}^2\le1.
\label{vis-ineq}
\end{equation}

Any jointly diagonalizable family of detector states satisfies
Eq.~\eqref{classical}, and therefore the visibility inequality
\eqref{vis-ineq} in the symmetric interferometer. Consequently, any
violation of Eq.~\eqref{vis-ineq} witnesses basis-independent coherence
by ruling out a jointly diagonalizable description.

\subsection{Maximal quantum violation for three states}

For $n=3$, the optimal overlap configuration yielding $S=5/4$ was
previously identified in Ref.~\cite{Galvao2020}. We include a
self-contained derivation because it provides the conceptual starting
point for the general $n$-cycle analysis developed below.

\begin{theorem}
\label{th:maximal}
For any three pure detector states with overlaps
$r_{ij}=|\langle d_i|d_j\rangle|^2$,
\[
S \equiv r_{12}+r_{23}-r_{13} \le \frac54.
\]
The bound is tight and is achieved by
$|d_2\rangle=|0\rangle$,
$|d_{1,3}\rangle=\tfrac{\sqrt3}{2}|0\rangle\pm\tfrac12|1\rangle$,
giving $r_{12}=r_{23}=3/4$, $r_{13}=1/4$.
\end{theorem}

\begin{proof}
For fixed $r_{12}$ and $r_{23}$, minimizing $r_{13}$ maximizes
$S=r_{12}+r_{23}-r_{13}$. From Appendix~\ref{app:gram},
\[
r_{13}\ge
\begin{cases}
\left(\sqrt{r_{12}r_{23}}-\sqrt{(1-r_{12})(1-r_{23})}\right)^2, & r_{12}+r_{23}>1,\\
0, & r_{12}+r_{23}\le 1.
\end{cases}
\]

If $r_{12}+r_{23}\le1$, then $S\le1$. Hence the optimum requires
$r_{12}+r_{23}>1$ and
\[
S\le r_{12}+r_{23}-(\sqrt{r_{12}r_{23}}-
\sqrt{(1-r_{12})(1-r_{23})})^2.
\]

Set $r_{12}=\cos^2\beta$, $r_{23}=\cos^2\gamma$
($\beta,\gamma\in[0,\pi/2]$). Then
$S\le\cos^2\beta+\cos^2\gamma-\cos^2(\beta+\gamma)=:F(\beta,\gamma)$.

Let $u=\beta+\gamma$, $\delta=\beta-\gamma$. Since $r_{12}+r_{23}>1$,
\[
\cos^2\beta+\cos^2\gamma
=
1+\cos u\cos\delta>1.
\]
Because $\beta,\gamma\in[0,\pi/2]$, we have
$\cos\delta\ge0$. Since $\cos u\cos\delta>0$, it follows that
$\cos u>0$, and hence
\[
u<\frac{\pi}{2}.
\]

Using
$\cos2\beta+\cos2\gamma=2\cos u\cos\delta$,
\[
F=\tfrac12+\cos u\cos\delta-\tfrac12\cos2u.
\]

Since $\cos u>0$, for fixed $u$ the quantity $F$ is maximized at
$\delta=0$. Therefore
\[
S\le 1+\cos u-\cos^2u
=\frac54-\left(\cos u-\frac12\right)^2
\le\frac54.
\]
Equality requires $u=\pi/3$ and $\delta=0$, giving
$\beta=\gamma=\pi/6$. Hence $r_{12}=r_{23}=3/4$, and from the sharp bound
$r_{13}=1/4$. Therefore
\[
S_{\max}=\tfrac34+\tfrac34-\tfrac14=\tfrac54.
\]
\end{proof}

Since the bound holds in every finite dimension and is attained by a
qubit realization, the global maximum over finite-dimensional pure-state realizations is achieved already
in dimension two; higher-dimensional pure-state systems cannot yield larger values of $S$.

\subsection{Three-path experimental realization}
\label{subsec:three-path}

A feasible experimental implementation can use three which-path detector states
encoded in photon polarization:
\begin{equation}
|d_1\rangle
=
\frac{\sqrt3}{2}|H\rangle
+
\frac12|V\rangle,
\quad
|d_2\rangle
=
|H\rangle,
\quad
|d_3\rangle
=
\frac{\sqrt3}{2}|H\rangle
-
\frac12|V\rangle.
\end{equation}
These states realize the optimal configuration of
Theorem~\ref{th:maximal} and saturate the quantum bound $S=5/4$.

In a symmetric three-path interferometer, the corresponding visibilities
are
\[
\mathcal{V}_{12}
=
\mathcal{V}_{23}
=
\frac{\sqrt3}{2}
\approx0.866,
\qquad
\mathcal{V}_{13}
=
\frac12.
\]
Consequently,
\begin{equation}
\mathcal{V}_{12}^2
+
\mathcal{V}_{23}^2
-
\mathcal{V}_{13}^2
=
\frac34+\frac34-\frac14
=
\frac54
>
1,
\end{equation}
demonstrating a maximal violation of the classical visibility bound.

\subsection{Uniform visibility reduction}

In realistic interferometers, imperfect mode overlap, decoherence, and
alignment errors reduce the observed two-path visibility. We model this
by a uniform efficiency factor $0<\eta\le1$, assuming
\[
\mathcal{V}_{ij}^{\mathrm{exp}}
=
\eta\,\mathcal{V}_{ij}
\]
for all path pairs.

The uniform reduction is a convenient benchmark when the setup is symmetric
and losses are approximately equal for all path pairs. If the losses are
pair-dependent, the calibrated treatment in Sec.~\ref{subsec:calib} should be
used instead.

Since each term in the cycle expression is quadratic in the visibility,
the experimentally observed value becomes
\[
S_{\mathrm{exp}}
=
\eta^2 S.
\]
Substituting into Eq.~\eqref{vis-ineq}, a violation occurs whenever
\[
\eta^2 S > 1,
\]
or equivalently,
\[
\eta>\frac{1}{\sqrt S}.
\]

For the optimal three-state configuration,
\[
\eta
>
\frac{2}{\sqrt5}
\approx0.894.
\]
This threshold provides an indicative benchmark for the visibility
performance achievable on integrated photonic platforms~\cite{Matthews2009,Heo2025}.

\subsection{Calibrated treatment of pair-dependent visibility loss}
\label{subsec:calib}

The uniform reduction model above provides a simple benchmark for the
required visibility. In an experiment, the visibility reduction may depend
on the path pair. Suppose that a calibration gives
\[
\mathcal{V}_{ij}^{\rm exp}=\eta_{ij}\mathcal{V}_{ij},
\]
where $\eta_{ij}$ is the calibrated visibility-transfer factor for pair
$(i,j)$. For known path amplitudes, the ideal overlap can then be obtained
from Eq.~\eqref{eq:vis-general}. In the equal-amplitude case,
\[
r_{ij}=
\left(\frac{\mathcal{V}_{ij}^{\rm exp}}{\eta_{ij}}\right)^2.
\]
The calibrated values can then be substituted into $S_n$. The uncertainties
in the measured visibilities and in the calibration factors $\eta_{ij}$
should be propagated to the resulting uncertainty in $S_n$. Since the
closing term enters with a negative sign, pair-dependent visibility loss
cannot in general be treated as a simple uniform reduction of the whole
cycle expression. A violation should therefore be assessed using the
uncertainty interval of the calibrated $S_n$ rather than from the uncorrected
visibilities alone. For a quantitative experimental claim, the violation
should be reported together with a confidence interval (for example, a
95\% confidence interval) that excludes the classical bound.

\section{Connection to preparation noncontextuality}
\label{sec:lsss}

The interferometric visibility measurements can be interpreted within the
framework of preparation noncontextuality. While the overlap inequalities of
Refs.~\cite{Galvao2020,Wagner2024} are known to witness preparation
contextuality in principle, evaluating the required overlaps generally relies
on state tomography or dedicated overlap-estimation techniques such as SWAP
tests. Here we show how a standard multi-path interferometer can provide the
corresponding quantum pairwise overlaps through visibility measurements,
under the ideal-model and calibration assumptions stated above. The
contextuality conclusion additionally requires independent implementation
and verification of the stated operational equivalences. Generalized
contextuality tests that do not rely on unphysical idealizations have been
demonstrated experimentally~\cite{Mazurek2016}.

The detector states $\{|d_i\rangle\}$ prepared by the interferometer
define a family of operational preparation procedures $\mathcal{P}_i$,
represented in quantum theory by the density operators
$|d_i\rangle\langle d_i|$. For each detector state,
consider the binary projective measurement
\[
\mathcal{M}_i
=
\Bigl\{
|d_i\rangle\langle d_i|,
\,
\mathbb{I}-|d_i\rangle\langle d_i|
\Bigr\},
\]
where, for qubit states, the second outcome corresponds to the rank-one
projector $|d_i^\perp\rangle\langle d_i^\perp|$.

Within the ideal pure-state quantum model, the pairwise overlap
$r_{ij}=|\langle d_i|d_j\rangle|^2$ is the probability of obtaining the first
outcome of $\mathcal{M}_i$ when the system is prepared according to
$\mathcal{P}_j$. In the symmetric interferometer,
Eq.~\eqref{eq:vis-sym} gives $r_{ij}=\mathcal{V}_{ij}^{\,2}$, so the measured
visibilities determine the quantum overlaps without full state tomography
or multi-copy measurements. This identification is model-dependent and, by
itself, is not a theory-independent operational definition of the pairwise
overlaps used in a generalized-noncontextuality test.

To connect to preparation noncontextuality, note that for qubit states
the associated outcome preparations
$\mathcal{P}_{i,0}=|d_i\rangle\langle d_i|$ and
$\mathcal{P}_{i,1}=|d_i^\perp\rangle\langle d_i^\perp|$ satisfy
\[
\frac12\mathcal{P}_{i,0} + \frac12\mathcal{P}_{i,1}
= \frac{\mathbb{I}}{2},
\]
independently of $i$. Within the ideal qubit quantum model, these
equal-weight mixtures are represented by $\mathbb{I}/2$. For a
preparation-contextuality test, however, the corresponding operational
equivalence must be independently established from experimentally accessible
statistics. The relevant \emph{prepare-and-measure} conditions also require the
measurement $\mathcal M_i$ to give the first outcome with probability one for
$\mathcal P_{i,0}$ and probability zero for $\mathcal P_{i,1}$ in the ideal
case. In an experiment, these probabilities may instead be verified within
appropriate finite uncertainties. Preparation noncontextuality~\cite{Spekkens2005}
requires that operationally equivalent preparation procedures be represented by
identical epistemic states in any ontological model, which constrains the
corresponding overlap quantities to lie within the classical overlap polytope
derived in Refs.~\cite{Galvao2020,Wagner2024}.
For three states, the nontrivial facet inequalities are
$r_{12}+r_{23}-r_{13}\le1$ and its cyclic permutations, which
translate via Eq.~\eqref{eq:vis-sym} to the visibility
inequality~\eqref{vis-ineq}.

Thus the same numerical violation has two levels of interpretation.
Under the ideal quantum mapping $r_{ij}=\mathcal V_{ij}^2$, it rules out a
jointly diagonalizable quantum realization~\cite{Designolle2021}. If, in
addition, the full \emph{prepare-and-measure} operational conditions required for
the preparation-noncontextual construction are independently implemented and
verified, the corresponding overlap inequality rules out
preparation-noncontextual ontological models. Pairwise visibilities provide a
convenient route to the quantum overlaps, but they do not replace this
independent operational verification.

\section{Extension to the $n$-cycle}
\label{sec:tight-n}

The three-state inequality of Theorem~\ref{th:maximal} is the first
nontrivial member of a hierarchy of overlap-cycle inequalities. We next
extend the three-state result to arbitrary cycle length $n$. For arbitrary
cycle length $n\ge3$, define the $n$-cycle overlap expression
\begin{equation}
S_n \equiv \sum_{i=1}^{n-1} r_{i,i+1} - r_{1n},
\label{eq:Sn-def}
\end{equation}
where $r_{ij}=|\langle d_i|d_j\rangle|^2$. Equivalent cycle inequalities
are obtained by cyclic permutations of the detector labels.

We denote the overlap functional by $S_n$ and its quantum maximum by
$S_n^{\max}$.

Under the equal-amplitude, pure-state, ideal-interferometry assumptions
of Sec.~\ref{sec:interferometer}, the corresponding {\em classical visibility inequality} is
\begin{equation}
\sum_{i=1}^{n-1}\mathcal V_{i,i+1}^{\,2}
-\mathcal V_{1n}^{\,2}
\le n-2.
\label{vis-ineq-general}
\end{equation}
This inequality is satisfied by any family of detector states that can be
simultaneously diagonalized in a common basis. A violation therefore shows
that no such common basis exists, witnessing basis-independent coherence.

For jointly diagonalizable (classical) states, iterating the three-state
inequality $r_{ij}+r_{jk}-r_{ik}\le1$ along the cycle gives the classical
bound $S_n\le n-2$ (Proposition~\ref{prop:classical}). The first quantum violation of the classical bound occurs at $n=3$. We now determine the exact quantum
maximum of $S_n$ for arbitrary $n\ge3$ and prove that it is already
attained by qubit states.

\begin{proposition}[Classical bound for the $n$-cycle]
\label{prop:classical}
For any jointly diagonalizable family of states,
\[
S_n := \sum_{i=1}^{n-1} r_{i,i+1} - r_{1n} \le n-2.
\]
\end{proposition}

\begin{proof}
We proceed by induction on $n$. For $n=3$, the claim is precisely the
three-state overlap inequality $r_{12}+r_{23}-r_{13}\le1$. Assume
$S_n\le n-2$. Writing
$S_{n+1}=S_n+(r_{1n}+r_{n,n+1}-r_{1,n+1})$, the bracketed term is
bounded by $1$ by the three-state inequality. Hence
$S_{n+1}\le (n-2)+1=n-1$, completing the induction.
\end{proof}

The classical $n$-cycle bound is obtained by iterating the three-state inequality; the quantum maximum, however, does not follow from a simple iteration of the three-state quantum bound and requires the separate geometric optimisation developed below.

\begin{lemma}[Fixed-sum extremization]
\label{lem:fixedT}
Let $m \ge 2$ be an integer and $T \in [0,\pi/2]$. Then
\[
\max_{\substack{\alpha_i \in [0,\pi/2]\\ \sum_{i=1}^m \alpha_i = T}}
\sum_{i=1}^m \cos^2 \alpha_i
=
m\cos^2\!\left(\frac{T}{m}\right),
\]
and the uniform configuration $\alpha_i=T/m$ is a maximizer.
It is unique except in the boundary case $m=2$ and $T=\pi/2$,
where every feasible pair $(\alpha,\pi/2-\alpha)$ is a maximizer.
\end{lemma}

\begin{proof}
Set $x_i=2\alpha_i\in[0,\pi]$. Then $\sum_{i=1}^m x_i = 2T \le \pi$.
Since
\[
\sum_{i=1}^m\cos^2\alpha_i
= \frac{m}{2} + \frac{1}{2}\sum_{i=1}^m\cos x_i,
\]
it suffices to maximize $F(x):=\sum_{i=1}^m\cos x_i$ subject to
$x_i\in[0,\pi]$ and $\sum_i x_i = 2T$.

If $x_a\neq x_b$ and $x_a+x_b=s<\pi$, averaging them to $(s/2,s/2)$
preserves the constraints. With $d=x_a-x_b$, the sum-to-product identity gives $\cos x_a+\cos x_b=2\cos(s/2)\cos(d/2)$. Since $0<|d|\le s<\pi$, we have $\cos(s/2)>0$ and $0<\cos(d/2)<1$, so $\cos x_a+\cos x_b<2\cos(s/2)$. Thus averaging strictly increases $F$.

Now first suppose $m\ge3$. Let $(x_1,\ldots,x_m)$ be a non-uniform maximizer.
No unequal pair can sum to $<\pi$, else averaging would increase $F$.
Hence every unequal pair must satisfy $x_a+x_b\ge\pi$. Pick
$x_a\neq x_b$; then, since all coordinates are nonnegative,
\[
\pi\le x_a+x_b\le\sum_i x_i = 2T \le \pi,
\]
so $x_a+x_b = 2T = \pi$ and all remaining coordinates vanish. Because
$m\ge3$, at least one zero coordinate is available. If either $x_a$ or
$x_b$ lies in $(0,\pi)$, pairing it with that zero coordinate produces an
unequal pair with sum $<\pi$, a contradiction. Hence the only remaining
non-uniform candidate is $(\pi,0,\ldots,0)$ up to permutation, and
necessarily $2T = \pi$.

For this candidate,
\[
F = m-2.
\]
The corresponding uniform configuration has
\[
x_i = \frac{\pi}{m},
\qquad
F = m\cos\!\left(\frac{\pi}{m}\right).
\]

Using $1-\cos x < x^2/2$ for $x>0$,
\begin{align*}
m\cos(\pi/m)-(m-2)
&= 2 - m\bigl(1-\cos(\pi/m)\bigr)
> 2 - \frac{\pi^2}{2m} > 0
\end{align*}
for all $m \ge 3$. Hence the uniform configuration gives a strictly
larger value than the non-uniform candidate, so it is the unique
maximizer for $m\ge 3$.

For $m=2$, strict improvement under averaging holds whenever $T<\pi/2$.
If $T=\pi/2$, then $\alpha_2=\pi/2-\alpha_1$ and
\[
\cos^2\alpha_1+\cos^2\alpha_2 = \cos^2\alpha_1+\sin^2\alpha_1 = 1,
\]
so every feasible pair attains the optimal value $1=2\cos^2(\pi/4)$.
Thus, for $m\ge3$ the uniform configuration is the unique maximizer.
For $m=2$, the uniform configuration is a maximizer, unique except at
$T=\pi/2$ where the objective is constant on the feasible set.
\end{proof}

\begin{lemma}[Monotonicity of the maximal sum]
\label{lem:monotone}
For $m\ge2$, define
\[
g(T) = \max_{\substack{\alpha_i \in [0,\pi/2] \\
\sum_{i=1}^m \alpha_i = T}} 
\sum_{i=1}^m \cos^2\alpha_i,
\qquad T\in[0,m\pi/2].
\]
Then $g(T)$ is nonincreasing on $[0,m\pi/2]$.
\end{lemma}

\begin{proof}
Let $T_2>T_1$ and let $\{\alpha_i\}$ be a maximizer for $g(T_2)$.
Set $\delta=T_2-T_1>0$. Since
\[
\delta\le T_2=\sum_{i=1}^m\alpha_i,
\]
choose $r_i\in[0,\alpha_i]$ with $\sum_{i=1}^m r_i=\delta$, and set
$\beta_i:=\alpha_i-r_i$. Then $\beta_i\in[0,\alpha_i]$ and
\[
\sum_{i=1}^m\beta_i
=
\sum_{i=1}^m\alpha_i
-
\sum_{i=1}^m r_i
=
T_2-\delta
=
T_1.
\]

Since $(\beta_1,\ldots,\beta_m)$ is feasible for the optimization problem
defining $g(T_1)$, we have
\[
g(T_1) \ge \sum_{i=1}^m \cos^2\beta_i.
\]

Because $\cos^2$ is decreasing on $[0,\pi/2]$, we have
$\cos^2\beta_i \ge \cos^2\alpha_i$ for each $i$. Hence
\[
\sum_{i=1}^m \cos^2\beta_i
\ge
\sum_{i=1}^m \cos^2\alpha_i
=
g(T_2).
\]

Combining the two inequalities gives $g(T_1) \ge g(T_2)$.
\end{proof}

\begin{theorem}[Exact quantum maximum of the $n$-cycle for pure states]
\label{thm:main}
For every integer $n\ge3$ and every finite-dimensional realization of
pure detector states,
\[
S_n = \sum_{i=1}^{n-1} r_{i,i+1} - r_{1n}
\le n\cos^2\!\left(\frac{\pi}{2n}\right)-1.
\]
The bound is tight.
\end{theorem}

\begin{proof}
We define the Fubini--Study angles
\[
\begin{aligned}
\alpha_i &= \arccos|\langle d_i|d_{i+1}\rangle| \in [0,\pi/2]
\quad (i=1,\dots,n-1),\\
\alpha_{1n} &= \arccos|\langle d_1|d_n\rangle| \in [0,\pi/2],
\end{aligned}
\]
where $\alpha_{ij}=\arccos|\langle d_i|d_j\rangle|$ is the
Fubini--Study distance between pure
states~\cite{Wootters1981,BengtssonZyczkowski2017}.

Since the Fubini--Study distance is a metric on projective Hilbert
space~\cite{AnandanAharonov1990,BengtssonZyczkowski2017},
the triangle inequality yields
\[
\alpha_{1n} \le \sum_{i=1}^{n-1} \alpha_i \equiv T.
\]

Geometrically, $T$ is the total Fubini--Study length of the polygonal
chain $d_1\to d_2\to\cdots\to d_n$.

Moreover, $0\le\alpha_i\le\pi/2$ for all $i$, so
\[
0\le T\le \frac{(n-1)\pi}{2}.
\]

Using
$|\langle d_i|d_{i+1}\rangle|^2=\cos^2\alpha_i$
and
$|\langle d_1|d_n\rangle|^2=\cos^2\alpha_{1n}$,
we obtain
\begin{equation}
S_n = \sum_{i=1}^{n-1} \cos^2\alpha_i - \cos^2\alpha_{1n}.
\label{eq:Sn-angle}
\end{equation}

Since $\cos^2 x$ is decreasing on $[0,\pi/2]$ but not on all of $[0,\pi]$,
we consider separately the cases $T\le\pi/2$ and $T>\pi/2$.

\subsection*{Case 1: $T \le \pi/2$}

For $T\le\pi/2$, the inequality $\alpha_{1n}\le T$ implies
$\cos^2\alpha_{1n}\ge\cos^2T$. Hence from \eqref{eq:Sn-angle},
\[
S_n \le \sum_{i=1}^{n-1}\cos^2\alpha_i-\cos^2T.
\]

Applying Lemma~\ref{lem:fixedT} with $m = n-1$ gives
\[
\sum_{i=1}^{n-1} \cos^2\alpha_i
\le (n-1)\cos^2\!\left(\frac{T}{n-1}\right),
\]
so
\[
S_n \le H(T),
\qquad
H(T):=(n-1)\cos^2\!\left(\frac{T}{n-1}\right)-\cos^2 T.
\]

To maximize $H(T)$ on $[0,\pi/2]$, set
\[
\phi = \frac{T}{n-1}, \qquad
0\le\phi\le\frac{\pi}{2(n-1)},
\]
and define
\[
K(\phi) := (n-1)\cos^2\phi - \cos^2\!\big((n-1)\phi\big).
\]
Then
\[
K'(\phi) = 2(n-1)\cos(n\phi)\sin((n-2)\phi).
\]

The stationary points satisfy
$\cos(n\phi)\sin((n-2)\phi)=0$.
The equation $\cos(n\phi)=0$ yields
$\phi=(2k+1)\pi/(2n)$, of which only
$\phi=\pi/(2n)$ lies in the open interval $(0,\pi/(2(n-1)))$.
The equation $\sin((n-2)\phi)=0$ yields
$\phi=k\pi/(n-2)$, whose first positive solution
$\pi/(n-2)$ lies outside the interval $(0,\pi/(2(n-1)))$.
Since $0<\phi<\pi/[2(n-1)]$ implies $(n-2)\phi < \pi/2$, the factor
$\sin((n-2)\phi)$ is strictly positive throughout the interior of the
interval. Hence the sign of $K'(\phi)$ is determined entirely by
$\cos(n\phi)$, and $K'$ changes from positive to negative at
$\phi=\pi/(2n)$. Therefore $\phi=\pi/(2n)$ is the only interior
stationary point and is a strict local maximum.
The boundaries of the interval are $\phi=0$ and $\phi=\pi/[2(n-1)]$.

Evaluating $K$ at these three candidates:
\begin{align*}
K(0) &= n-2,\\
K\!\left(\frac{\pi}{2n}\right)
&= n\cos^2\!\left(\frac{\pi}{2n}\right)-1 \equiv M_1,\\
K\!\left(\frac{\pi}{2(n-1)}\right)
&= (n-1)\cos^2\!\left(\frac{\pi}{2(n-1)}\right) \equiv M_2.
\end{align*}

We first compare $M_1$ with $K(0)$:
\[
M_1-(n-2)=1-n\sin^2\!\left(\frac{\pi}{2n}\right)
>1-\frac{\pi^2}{4n}\ge 1-\frac{\pi^2}{12}>0,
\]
so $M_1>K(0)$. Appendix~\ref{app:boundary} proves $M_1>M_2$.

By the extreme value theorem, a global maximizer of $K$ on the closed
interval $[0,\pi/(2(n-1))]$ must occur either at an endpoint or at an
interior stationary point. Since $\phi=\pi/(2n)$ is the only interior
stationary point and its value exceeds both endpoints, it is the global
maximizer. Hence
\[
\max_{\phi}K(\phi)=K\!\left(\frac{\pi}{2n}\right)=M_1.
\]
Consequently $H(T)\le M_1$ for all $T\le\pi/2$, and hence
$\max_{T\le\pi/2} S_n \le M_1$.

\subsection*{Case 2: $T > \pi/2$}

In this region we cannot use the bound $\cos^2\alpha_{1n}\ge\cos^2 T$,
since $\cos^2 x$ is not monotonic on $[0,\pi]$. Instead, using
$\cos^2\alpha_{1n}\ge0$ in Eq.~\eqref{eq:Sn-angle} gives
\[
S_n \le \sum_{i=1}^{n-1} \cos^2\alpha_i.
\]

Let
\[
g(T) := \max_{\substack{\alpha_i\in[0,\pi/2]\\
\sum_{i=1}^{n-1}\alpha_i = T}}
\sum_{i=1}^{n-1} \cos^2\alpha_i.
\]
Then $S_n \le \sum_{i=1}^{n-1}\cos^2\alpha_i \le g(T)$.
By Lemma~\ref{lem:monotone}, $g(T) \le g(\pi/2)$ for $T > \pi/2$, and
Lemma~\ref{lem:fixedT} yields
\[
g(\pi/2) = (n-1)\cos^2\!\left(\frac{\pi}{2(n-1)}\right) =: M_2.
\]
Hence $S_n \le M_2$.

Since $M_2 < M_1$ (Appendix~\ref{app:boundary}), no global maximizer lies
in this region. Combining with Case~1 gives $S_n \le M_1$.

\subsection*{Tightness}
The bound is saturated by the $n$ coplanar qubit states
\[
\ket{d_k} = \cos\theta_k\ket{0} + \sin\theta_k\ket{1},\qquad
\theta_k = \frac{(k-1)\pi}{2n},\quad k=1,\ldots,n.
\]
For these states,
\[
r_{i,i+1} = \cos^2\!\left(\frac{\pi}{2n}\right),\qquad
r_{1n} = \sin^2\!\left(\frac{\pi}{2n}\right).
\]
Therefore
\[
S_n = (n-1)\cos^2\!\left(\frac{\pi}{2n}\right)-\sin^2\!\left(\frac{\pi}{2n}\right)
= n\cos^2\!\left(\frac{\pi}{2n}\right)-1 = M_1.
\]

For these states,
\[
|\langle d_i|d_j\rangle|
=
\cos(\theta_j-\theta_i),
\]
so the Fubini--Study distance between successive states is $\pi/(2n)$.
Hence the states are equally spaced along a common Fubini--Study
geodesic, corresponding to an arc of a great circle on the Bloch sphere.
The total Fubini--Study length is $(n-1)\pi/(2n)$, equivalently a
Bloch-sphere angular length of $(n-1)\pi/n$.

The open-arc geometry reflects the asymmetric sign structure of $S_n$:
the positive adjacent terms favour small successive distances, whereas the
negative closing term favours a large endpoint distance. In the maximizing
regime $T\le\pi/2$, the Fubini--Study triangle inequality allows the endpoint
distance to be maximized when the chain follows a common minimizing geodesic;
the fixed-sum extremization then selects equal successive distances. This
explains why the saturating construction is an open arc rather than a closed
regular cycle. We establish this saturating family but do not claim a complete
uniqueness classification of all maximizers.

The upper bound is dimension-independent and is attained by the above
qubit construction. Therefore
\[
S_n^{\max} = n\cos^2\!\left(\frac{\pi}{2n}\right)-1.
\]
\end{proof}

The mixed-state case can be treated in the same way when the overlap
functional is defined by the Hilbert--Schmidt overlap
$r_{ij}=\operatorname{Tr}(\rho_i\rho_j)$. The functional
\[
S_n=\sum_{i=1}^{n-1}\operatorname{Tr}(\rho_i\rho_{i+1})
-\operatorname{Tr}(\rho_1\rho_n)
\]
is linear in each density operator separately. Thus, fixing all states except
$\rho_i$, the expression is linear in $\rho_i$. Since the finite-dimensional
density-operator sets are compact and $S_n$ is continuous, a global maximum
is attained. Using the spectral decomposition
\[
\rho_i=\sum_k p_k|\psi_k\rangle\langle\psi_k|,
\]
linearity gives
\[
S_n(\rho_i)
=
\sum_k p_k
S_n(|\psi_k\rangle\langle\psi_k|)
\le
\max_k S_n(|\psi_k\rangle\langle\psi_k|).
\]
Thus, for fixed values of the other density operators, a pure state can
always be chosen with a value of $S_n$ at least as large as that of
$\rho_i$. Starting from a maximizing set of states and applying this
replacement successively to $\rho_1,\ldots,\rho_n$ yields a maximizing set
of pure states. Therefore the pure-state bound above also gives the
mixed-state maximum:
\[
S_n^{\max}=n\cos^2\!\left(\frac{\pi}{2n}\right)-1.
\]
Thus finite-dimensional mixed states do not improve the exact optimum,
although their interferometric visibility need not equal the
Hilbert--Schmidt overlap.

Representative saturating qubit configurations for $n=3$ and $n=4$
are shown in Fig.~\ref{fig:bloch}, while
Table~\ref{tab:tight-bounds} summarizes the classical bounds, exact
quantum maxima, and corresponding visibility thresholds.

\begin{figure}[t]
\centering
\includegraphics[width=0.75\columnwidth]{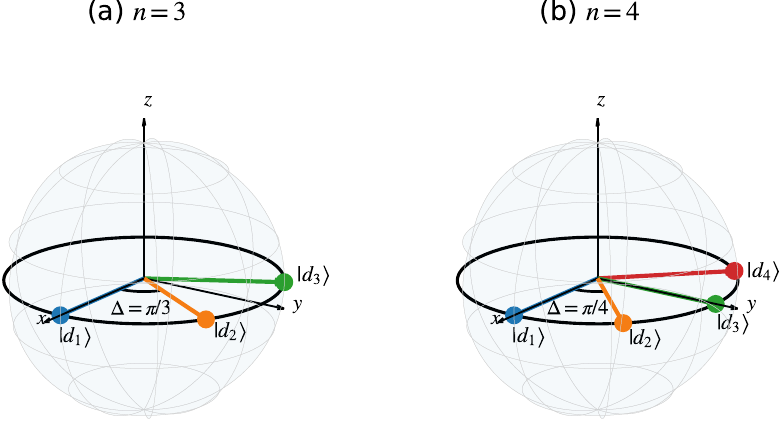}
\caption{Saturating qubit configurations for the $n$-cycle overlap inequality.
(a) $n=3$: three coplanar detector states with uniform Bloch-sphere angular
separation $\Delta=\pi/3$. (b) $n=4$: four coplanar detector states with
uniform Bloch-sphere angular separation $\Delta=\pi/4$. The corresponding
Fubini--Study separation is $\Delta_{\rm FS}=\Delta/2$.}
\label{fig:bloch}
\end{figure}

\begin{table*}[t]
\centering
\small
\setlength{\tabcolsep}{3pt}
\renewcommand{\arraystretch}{1.15}
\begin{tabular}{ccccc}
\hline
$n$ & Classical & $S_n^{\max}$ & $\Delta_n$ & $\eta_{\min}$ \\
\hline
3 & 1 &
$5/4=1.250000$ &
$1/4=0.250000$ &
$2/\sqrt5=0.894427$ \\
4 & 2 &
$1+\sqrt2=2.414214$ &
$\sqrt2-1=0.414214$ &
$\sqrt{2/(1+\sqrt2)}=0.910180$ \\
5 & 3 &
$(17+5\sqrt5)/8=3.522542$ &
$(5\sqrt5-7)/8=0.522542$ &
$\sqrt{24/(17+5\sqrt5)}=0.922853$ \\
6 & 4 &
$2+3\sqrt3/2=4.598076$ &
$3\sqrt3/2-2=0.598076$ &
$\sqrt{8/(4+3\sqrt3)}=0.932700$ \\
\hline
\end{tabular}
\caption{Classical bounds, exact quantum maxima, quantum--classical gaps
$\Delta_n=S_n^{\max}-(n-2)$, and uniform visibility thresholds
$\eta_{\min}=\sqrt{(n-2)/S_n^{\max}}$ for the $n$-cycle overlap
inequality. Numerical values are shown to 6 decimal places.}
\label{tab:tight-bounds}
\end{table*}

\subsection{Asymptotic behavior}

For large $n$, expansion of the quantum maximum gives
\[
S_n^{\max} = n\cos^2\!\left(\frac{\pi}{2n}\right)-1
= n-1 - \frac{\pi^2}{4n} + O(n^{-3}).
\]
Thus the quantum advantage over the classical bound $n-2$ approaches a
constant:
\[
S_n^{\max} - (n-2) = 1 - \frac{\pi^2}{4n} + O(n^{-3})
\xrightarrow{n\to\infty} 1.
\]
A finite gap persists even for arbitrarily long cycles. However, the
visibility threshold scales as
\[
\eta_{\min} = \sqrt{\frac{n-2}{S_n^{\max}}}
= 1 - \frac{1}{2n} + O(n^{-2}),
\]
so near-perfect visibility is required to observe violations for
large $n$.

\section{Discussion and experimental feasibility}
\label{sec:discussion}

The overlap inequalities introduced by Galv\~{a}o and
Brod~\cite{Galvao2020} and further developed by Wagner \emph{et
al.}~\cite{Wagner2024} provide a basis-independent coherence witness
and, under operational equivalences, a preparation-noncontextuality
inequality. Under the pure-state interferometric assumptions stated above,
pairwise visibility measurements determine the pairwise quantum overlaps
entering these inequalities. When the relevant operational
equivalences are experimentally satisfied, a violation rules out
preparation-noncontextual models and certifies basis-independent
coherence. Related overlap-based coherence witnesses have also been
investigated using multiphoton indistinguishability tests~\cite{Giordani2021}.

Pairwise visibilities are obtained by opening one pair of paths at a
time and extracting the corresponding fringe contrast, provided the path
amplitudes are known and stable. To determine all pairwise overlaps,
different path-pair selection configurations are required, but no full
state tomography or multicopy overlap measurement is needed. The same
detector-state preparation should be maintained across all configurations,
and its consistency should be checked experimentally. Recent experimental
work has demonstrated coherence and dimension certification based on
pairwise overlap measurements in a programmable universal integrated
photonic processor, and has also demonstrated a contextual advantage in a
quantum-interrogation task~\cite{Giordani2023}. These results support the
practical feasibility of overlap-based certification in scalable photonic
platforms.

A feasible experimental implementation can use the optimal three-path
configuration described in Sec.~\ref{subsec:three-path}, with
polarization-encoded detector states achieving $S_3=5/4$. This requires
visibility $\eta>2/\sqrt5\approx0.894$. A scalable four-path
implementation uses detector polarization states at angles $0^\circ$,
$22.5^\circ$, $45^\circ$, and $67.5^\circ$ relative to the horizontal
axis, requiring $\eta>\sqrt{2/(1+\sqrt{2})}\approx0.910$ for a violation
(see Table~\ref{tab:tight-bounds} for additional cycle lengths). These
values provide indicative benchmarks for high-visibility interference
achievable on integrated photonic platforms~\cite{Matthews2009,Heo2025}.
For the contextuality interpretation, the orthogonal preparations
$|d_i^\perp\rangle$ can be generated using the same polarization-control
optics and should be independently verified, for example by projective
measurements confirming the required orthogonality within the experimental
uncertainty. Such orthogonality verification is necessary but does not
by itself establish the full operational equivalences required for a
theory-independent generalized-noncontextuality test; the main contextuality
discussion in Sec.~\ref{sec:lsss} makes the required operational
assumptions explicit.

A proof-of-principle experiment would require stable path populations,
good mode matching, and sufficient phase stability; the more detailed
calibration and statistical analysis are given in Sec.~\ref{subsec:calib}.

\section{Conclusion}
\label{sec:conclusion}

The three-state maximum $S_3^{\max}=5/4$ is known from previous work
~\cite{Galvao2020}. The present work extends this result to arbitrary cycle
length and determines the exact quantum maximum,
\[
S_n^{\max}
=
n\cos^2\!\left(\frac{\pi}{2n}\right)-1,
\]
attained by coplanar qubit states equally spaced along a common Fubini--Study
geodesic, corresponding to an arc of a great circle on the Bloch sphere.
The same maximum holds for finite-dimensional mixed detector states when the
overlap functional is defined by the Hilbert--Schmidt overlap
$r_{ij}=\operatorname{Tr}(\rho_i\rho_j)$. This statement concerns the overlap
functional itself; for mixed detector states, the interferometric visibility
need not equal the Hilbert--Schmidt overlap.

The exact bound also has a simple asymptotic consequence: the quantum
advantage over the classical bound approaches $1$ as $n\to\infty$. Under
the pure-state, equal-amplitude interferometric assumptions stated above,
the pairwise quantum overlaps entering these inequalities can be inferred
from pairwise visibility measurements. This does not require full state
tomography or multicopy overlap estimation. A preparation-contextuality test
additionally requires independent experimental verification of the relevant
operational equivalences. In particular, the three-path visibility
inequality
\[
\mathcal{V}_{12}^2+\mathcal{V}_{23}^2-\mathcal{V}_{13}^2\le1
\]
admits a maximal quantum value of $5/4$, attained by pure qubit states.
When the independently verified operational equivalences required by
generalized noncontextuality are satisfied, violations rule out
preparation-noncontextual models. They also witness basis-independent
coherence by ruling out a jointly diagonalizable description.

Our results provide exact quantum benchmarks for the overlap-cycle hierarchy
and show how interference visibility can be used to access the pairwise
overlaps entering these inequalities under the stated calibration and
operational assumptions.

\appendix
\section{Gram matrix constraints and derivation of the overlap bound}
\label{app:gram}

For three pure detector states
$\{|d_1\rangle,|d_2\rangle,|d_3\rangle\}$, define the Gram matrix
\begin{equation}
G =
\begin{pmatrix}
1 & \langle d_1|d_2\rangle & \langle d_1|d_3\rangle \\
\langle d_2|d_1\rangle & 1 & \langle d_2|d_3\rangle \\
\langle d_3|d_1\rangle & \langle d_3|d_2\rangle & 1
\end{pmatrix}.
\end{equation}
A set of complex numbers $\{\langle d_i|d_j\rangle\}$ with
$\langle d_i|d_i\rangle = 1$ and
$\langle d_i|d_j\rangle = \overline{\langle d_j|d_i\rangle}$ corresponds
to the inner products of a family of pure quantum states if and only if
the Gram matrix $G$ is positive semidefinite (see, e.g.,
Ref.~\cite{HornJohnson2012}). For a $3\times3$ Gram matrix with unit
diagonal, the $2\times2$ principal minors are $1-r_{ij}$ and are
therefore nonnegative whenever $0\le r_{ij}\le1$. Hence, under the
physical constraints $0\le r_{ij}\le1$, positive semidefiniteness of
$G$ reduces to the single condition $\det G\ge0$.

Let $r_{ij} = |\langle d_i|d_j\rangle|^2$ and write the overlaps as
$\langle d_i|d_j\rangle = \sqrt{r_{ij}}\,e^{i\theta_{ij}}$. By suitable
rephasing of the states, two overlaps may be chosen real and positive,
leaving a single gauge-invariant phase:
\begin{equation}
\langle d_1|d_2\rangle = \sqrt{r_{12}},
\quad
\langle d_2|d_3\rangle = \sqrt{r_{23}},
\quad
\langle d_1|d_3\rangle = \sqrt{r_{13}}\,e^{i\varphi},
\end{equation}
with $\varphi = \theta_{13} - \theta_{12} - \theta_{23}$.

Evaluating the determinant yields
\begin{equation}
\det G = 1 + 2\sqrt{r_{12}r_{23}r_{13}}\cos\varphi
- r_{12} - r_{23} - r_{13}.
\end{equation}
Since the coefficient of $\cos\varphi$ is
$2\sqrt{r_{12}r_{23}r_{13}}\ge0$,
$\det G$ is nondecreasing in $\cos\varphi$.
Therefore, for fixed $r_{12}$, $r_{23}$, and $r_{13}$,
its largest value is attained when $\cos\varphi=1$.
Since admissibility requires $\det G\ge0$, the smallest admissible
value of $r_{13}$ is attained when $\varphi=0$.

Setting $\varphi=0$, the determinant condition becomes
\[
1 + 2\sqrt{r_{12}r_{23}r_{13}} - r_{12} - r_{23} - r_{13} \ge 0.
\]

For fixed $r_{12}$ and $r_{23}$, viewed as a function of $r_{13}$, the
left-hand side increases for $r_{13} < r_{12}r_{23}$, reaches a maximum
at $r_{13}=r_{12}r_{23}$, and decreases for $r_{13}>r_{12}r_{23}$.
The condition is therefore satisfied precisely between the two roots of
the corresponding equality. Consequently, the desired lower bound on
$r_{13}$ corresponds to the smaller root.

Defining $x = \sqrt{r_{13}} \ge 0$, the inequality reduces to
\begin{equation}
x^2 - 2\sqrt{r_{12}r_{23}}\,x + (r_{12}+r_{23}-1) \le 0.
\label{quad-ineq}
\end{equation}
Since the coefficient of $x^2$ is positive, the inequality is satisfied
for $x$ between the two real roots,
\begin{equation}
x_\pm = \sqrt{r_{12}r_{23}} \pm \sqrt{(1-r_{12})(1-r_{23})}.
\end{equation}
For $r_{12}+r_{23} > 1$, the lower root $x_-$ is positive and gives the
minimal admissible value of $x$, yielding the tight bound
\begin{equation}
r_{13} \ge \bigl( \sqrt{r_{12}r_{23}}
- \sqrt{(1-r_{12})(1-r_{23})} \bigr)^2.
\label{bound-strong}
\end{equation}
By contrast, when $r_{12}+r_{23} \le 1$, the constant term in
Eq.~\eqref{quad-ineq} is nonpositive, so $x=0$ satisfies the quadratic
constraint and therefore $r_{13}=0$ is admissible as the tight lower
bound.

Equality in Eq.~\eqref{bound-strong} is attained, for example, when
$\varphi = 0$ for three pure qubit states whose Bloch vectors lie on a
common great circle. In this case all pairwise overlaps may be chosen
simultaneously real and positive. The determinant constraint provides
the nontrivial ingredient in the optimization of $S$ (together with the
physical bounds $0 \le r_{ij} \le 1$); see also Ref.~\cite{Galvao2020}.

\section{Numerical consistency check}
\label{app:numerical}

As an independent numerical check of the analytical bound, we optimized
\[
S_n=
\sum_{i=1}^{n-1} |\langle d_i|d_{i+1}\rangle|^2
-|\langle d_1|d_n\rangle|^2
\]
over pure states $|d_i\rangle\in\mathbb C^d$ without imposing a qubit ansatz.
We used the SLSQP (Sequential Least Squares Programming) constrained
optimisation method, as implemented in SciPy's \texttt{minimize} function
~\cite{NocedalWright2006,SciPy2020}, with explicit unit-norm constraints
\[
\langle d_i|d_i\rangle=1
\]
for each state. The states were parameterised by their real and imaginary
parts. For each $(n,d)$ pair with $d=2,3,6,8$ and $n=3,4,5,6$, we performed
$20$ multistart runs from independently generated random initial states and
retained the largest value obtained.

In all cases the optimiser reported successful termination. The numerical
results are consistent with the analytical expression to the precision
shown in Table~\ref{tab:numcheck}; the small residual differences are
consistent with finite numerical optimisation tolerance and
floating-point arithmetic. This numerical check supports the analytical
proof but does not replace it; the global maximum is established by the
theorem above. The numerical values are listed below.

\begin{table}[h]
\centering
\small
\begin{tabular}{c|cccc}
\hline
$d$ & $n=3$ & $n=4$ & $n=5$ & $n=6$\\
\hline
2 & 1.2500000000 & 2.4142135624 & 3.5225424859 & 4.5980762114\\
3 & 1.2500000000 & 2.4142135624 & 3.5225424859 & 4.5980762114\\
6 & 1.2500000000 & 2.4142135624 & 3.5225424859 & 4.5980762114\\
8 & 1.2500000000 & 2.4142135624 & 3.5225424859 & 4.5980762114\\
\hline
\end{tabular}
\caption{Numerical optimization results for dimensions $d=2,3,6,8$
and cycle lengths $n=3,4,5,6$. Values are shown to the precision reported.}
\label{tab:numcheck}
\end{table}
The Python code used for the numerical checks is provided as Supplementary Material.

\section{Endpoint comparison}
\label{app:boundary}

We show
\[
K\!\left(\frac{\pi}{2n}\right)
>
K\!\left(\frac{\pi}{2(n-1)}\right)
\]
for all integers $n\ge3$.

Using the explicit values
\[
K\!\left(\frac{\pi}{2n}\right)
=
n\cos^2\!\left(\frac{\pi}{2n}\right)-1,
\]
\[
K\!\left(\frac{\pi}{2(n-1)}\right)
=
(n-1)\cos^2\!\left(\frac{\pi}{2(n-1)}\right),
\]
it suffices to prove
\[
(n-1)\cos^2\!\left(\frac{\pi}{2(n-1)}\right)
<
n\cos^2\!\left(\frac{\pi}{2n}\right)-1.
\]

Using $\cos^2x=(1+\cos2x)/2$, this is equivalent to
\[
(n-1)\cos\!\left(\frac{\pi}{n-1}\right)
-
n\cos\!\left(\frac{\pi}{n}\right)
<
-1.
\]

Define
\[
h(x)=x\cos\!\left(\frac{\pi}{x}\right),
\qquad x\ge2.
\]
Then the inequality becomes $h(n-1)-h(n)<-1$.

By the Mean Value Theorem, there exists $\xi\in(n-1,n)$ such that $h(n)-h(n-1)=h'(\xi)$.
Hence
\[
h(n-1)-h(n)=-h'(\xi).
\]

Computing the derivative,
\[
h'(x)=\cos\!\left(\frac{\pi}{x}\right)
+\frac{\pi}{x}\sin\!\left(\frac{\pi}{x}\right).
\]

Let $t=\pi/x$. For $x>2$, we have $t\in(0,\pi/2)$ and
\[
h'(x)=\cos t+t\sin t.
\]

Since
\[
\frac{d}{dt}\bigl(\cos t+t\sin t\bigr)
=
t\cos t
>
0
\qquad (0<t<\pi/2),
\]
the function $f(t):=\cos t+t\sin t$ is strictly increasing on
$(0,\pi/2]$. Moreover,
\[
\lim_{t\to0^+}f(t)=1.
\]
Hence
\[
f(t)>1
\qquad (0<t\le\pi/2).
\]
Substituting $t=\pi/x$, we obtain
\[
h'(x)>1
\qquad (x\ge2).
\]

Since $n\ge3$, we have $\xi\in(n-1,n)\subset(2,\infty)$.
Hence $h'(\xi)>1$, and therefore
\[
h(n-1)-h(n)=-h'(\xi)<-1,
\]
which proves the claim.

\bibliographystyle{RS}
\bibliography{references}

\end{document}